\documentclass[11pt]{article} 

\usepackage[utf8]{inputenc} 

\usepackage{geometry} 
\usepackage{graphicx} 
\usepackage{siunitx}

\usepackage{amssymb,amsmath,amsthm,mathtools}
\usepackage{subfig,caption}
\usepackage{hyperref}

\usepackage{booktabs} 
\usepackage{array,url} 
\usepackage{paralist} 
\usepackage{verbatim} 
\usepackage{subfig} 

\usepackage{fancyhdr} 
\usepackage{sectsty}
\allsectionsfont{\sffamily\mdseries\upshape} 

\usepackage[nottoc,notlof,notlot]{tocbibind} 
\usepackage[titles,subfigure]{tocloft} 

\usepackage{authblk}

\newtheorem{theorem}{Theorem}
\newtheorem{corollary}{Corollary}[theorem]
\newtheorem{proposition}{Proposition}[theorem]
\begin{document}
%
\title{\textbf{Tight Bounds on the Coefficients of\\ Consecutive $k$-out-of-$n$:$F$ Systems}}
%
%
\author[1,2]{Vlad-Florin Dr\u{a}goi}
\author[1] {Simon R. Cowell} 
\author[1] {Valeriu Beiu}

%
%
\affil[1]{Department of Mathematics and Computer Sciences \protect\\ Aurel Vlaicu University of Arad, Romania}
\affil[2]{LITIS, University of Rouen Normandie, France}
\date{}
\maketitle              

\begin{abstract} In this paper we compute the coefficients of the reliability polynomial of a consecutive-$k$-out-of-$n$:$F$ system, in Bernstein basis, using the generalized Pascal coefficients. Based on well-known combinatorial properties of the generalized Pascal triangle we determine simple closed formulae for the reliability polynomial of a consecutive system for particular ranges of $k$. Moreover, for the remaining ranges of $k$ (where we were not able to determine simple closed formulae), we establish easy to calculate sharp bounds for the reliability polynomial of a consecutive system.\\

\textbf{Keywords:} Consecutive systems, Generalized Pascal triangles, Bernstein basis, Reliability polynomial.
\end{abstract}
\section{Introduction}

A relatively hidden gem of network reliability is represented by the class of consecutive systems. They were introduced in 1980 as $r$-successive-out-of-$n$:$F$ systems \cite{K80}, before being aptly renamed consecutive-$k$-out-of-$n$:$F$ systems in 1981 \cite{CN1981}. Clearly, this type of redundancy scheme came reasonably late to the “reliability table,” i.e., almost 30 years after the majority-voting and the multiplexing concepts (both gate-level based reliability schemes) were introduced by von Neumann in January 1952\footnote{John von Neumann presented his work in five seminal lectures at the California Institute of Technology (Caltech) in January 1952. They are available, based on the notes taken by R. S. Pierce, at {\url{https://sites.google.com/site/michaeldgodfrey/vonneumann/vN_Caltech_Lecture.pdf}}}. A printed version of those lectures was published in April 1956 \cite{1956_vN}, followed in September 1956 by the introduction of the hammock networks by Moore and Shannon \cite{1956_MS_1} (the first device-level based reliability scheme). For more information on consecutive systems the interested reader should consult \cite{C2000}, \cite{E2010}, while it is worth mentioning that the associated probability problem was proposed and solved as early as 1718 by de Moivre \cite{dM1718} (see also \cite{DB15}), with the associated graphs being proven most reliable in the late 90's (see \cite{DCLLG_2004}).

Consecutive-$k$-out-of-$n$:$F$ systems belong to the class of device-level based reliability schemes (although “devices” might be quite complex entities), and are aimed at communications, as opposed to gate-level based reliability schemes which are targeting computations. Such systems can be abstracted as networks/ graphs, network reliability being a field pioneered by \cite{1956_MS_1} and which has significantly evolved ever since (see \cite{CC1997}, \cite{C87}, \cite{P2018}). The fundamental problems in network reliability are to determine: two-terminal, $k$-terminal, and all-terminal reliability of a network, and are all known to be very difficult in general (\#P-complete  \cite{K80}, \cite{MSW2018}, \cite{PB1983} \cite{1979_V}). That is why even the best algorithms are time consuming \cite{CK2008}, \cite{GGK_2016}, \cite{MSW2018}, and lower and upper bounds were investigated as efficient alternatives to exact but tedious computations. In the particular case of consecutive-$k$-out-of-$n$:$F$ systems, bounds have been reported staring from 1981 \cite{CN1981}, and improved over time (see \cite{CC1997}, \cite{C87}, \cite{DB15}, \cite{P2018}, \cite{P1986}). A 'midway path' forward is to bound the coefficients of the reliability polynomial \cite{C1990}, \cite{OW2002}, \cite{BCS2010}, and follow with the exact polynomial computations. All of these different approaches reveal wide trade-offs between accuracy and time-complexity. 

In this paper we are investigating a 'midway path' approach for the particular case of consecutive-$k$-out-of-$n$:$F$ systems, and we will show that most of the coefficients can be quite easily computed exactly, while only a handful of them are computationally demanding, but can be bounded by reasonably simple formulas.

\subsection{Consecutive systems}
A consecutive-$k$-out-of-$n$:$F$ system corresponds to a sequence of $n$ independent, identically distributed (i.i.d.) Bernoulli trials, with common probability of success $p$, in which the system itself is deemed to have failed if the sequence includes a run of at least $k$ consecutive failures, and to have succeeded, otherwise. The \emph{reliability} of the system is the probability $R(k,n;p)$ that it succeeds. We can write this probability as a homogeneous polynomial of degree $n$ in $p$ and $q$, where $q=1-p$, as follows:
\begin{equation}
\label{eq: reliability polynomial}
R(k,n;p)=\sum_{i=0}^n N_{n,k,i} p^i q^{n-i},
\end{equation}
where $N_{n,k,i}$ is the number of sequences of $n$ trials that include exactly $i$ successes, in which the longest consecutive run of failures has length strictly less than $k$.

\subsection{Standard multinomial coefficient}

A well known bins-and-balls counting problem that we consider here is the following. What is the number of ways in which $n$ identical balls can be distributed among a sequence of $i$ distinct bins, such that bins may be empty, and no bin may contain more than $k$ balls? The answer to this problem is given by the standard multinomial coefficient, denoted $\binom{i}{n}_k$. The algebraic description of $\binom{i}{n}_k$ is the following 
\begin{equation}
(1+z+z^2+\dots+z^k)^i=\sum\limits_{a\ge 0}^{}\binom{i}{a}_kz^a,\label{def:gen_pascal}
\end{equation} 
with $\binom{i}{a}_1$ the usual binomial coefficient and $\binom{i}{a}_k=0$ for $a>ik.$ 

More generally, such objects are also known to count the number of $A$-restricted compositions of an integer $n$ into $i$ parts. That is, the number of ways, $\binom{i}{n}_{(1)_{j \in A}}$, in which $n$ can be written as the sum of a sequence of $i$ integers drawn from a given subset $A \subseteq \{0, 1, \ldots \}$, with replacement (i.e., the order is important). When $A=\{0, \ldots, k\}$, we simply use the $\binom{i}{n}_k$ notation.

\section{Results}

\begin{theorem}
\label{thm: main result}
We have
\begin{equation}
\label{eq: bins and balls}
\begin{split}
N_{n,k,i}
&=[z^{n-i}](1+z+\cdots+z^{k-1})^{i+1} \\
&=[z^i]z^{n-(k-1)(i+1)} \left ( \frac{1 - z^k}{1 - x} \right )^{i+1},
\end{split}
\end{equation}
where $[z^t]f(z)$ denotes the coefficient of $z^t$ in the formal power series expansion of $f(z)$ in powers of $z$.
\end{theorem}
\begin{proof}
Our proof is a combinatorial one, that is, we show that two counting problems are identical. Fixing $k,n$ and $i$, consider a sequence of $n$ trials that includes exactly $i$ successes and in which all the runs of consecutive failures have length at most $k-1$. We may consider this sequence as a sequence of $i+1$ runs of consecutive failures of lengths between $0$ and $k-1$ inclusive, each consecutive pair of such runs separated by a single success, in which the total number of failures is $n-i$.
The number of such sequences, which is $N_{n,k,i}$, is therefore also the number of ways in which $n-i$ identical balls can be distributed among a sequence of $i+1$ distinct bins, such that bins may be empty, and no bin may contain more than $k-1$ balls.
The first equality in \eqref{eq: bins and balls} now follows directly from \eqref{def:gen_pascal}, 
and the second one 
follows from the identity $[z^{n-i}]f(z)=[z^i](z^n f(1/z))$.
\end{proof}
\subsection{Properties of the reliability polynomials}

\begin{theorem}
$N_{n,k,i}$ satisfy the following properties:
\begin{align}
N_{n,k,i}&=0, \forall i\leq i_{n,k}\triangleq\left\lfloor\frac{n-k+1}{k}\right\rfloor;\label{eq:Nk1}\\
&=\binom{n}{n-i},\forall i\ge n-k+1; \label{eq:Nk2}\\
&=\sum\limits_{j=0}^{\lfloor\frac{n-i}{k}\rfloor}(-1)^j\binom{i+1}{j}\binom{n-jk}{i}, \forall i \in \{i_{n,k}+1,\dots,n-k\}.\label{eq:Nk3}
\end{align}
\end{theorem}



\begin{corollary}The reliability polynomial of a consecutive-$k$-out-of-$n$:$F$ system  \begin{equation}R(k,n;p)=\sum\limits_{i=i_{n,k}+1}^{n}\binom{n}{i}p^iq^{n-i}-\sum\limits_{i=i_{n,k}+1}^{n-k}\sum\limits_{j=1}^{\lfloor\frac{n-i}{k}\rfloor}(-1)^{j+1}\binom{i+1}{j}\binom{n-jk}{i}p^iq^{n-i}.
\end{equation}
\end{corollary}

Equation \eqref{eq:Nk3} gives the full description of the coefficient $N_{n,k,i}$ regardless of the values of $k$ and $n$. However, by taking a closer look we can deduce simpler expressions for some sub-sets of $\{i_{n,k}+1,\dots,n-k\}.$

\begin{corollary}\label{cor:2}
\begin{align}
N_{n,k,i}&=\binom{n}{i}-(i+1)\binom{n-k}{i}, \forall i \in \{n-2k+1,\dots, n-k\};\label{eq:Nk4}\\
N_{n,k,i}&=\binom{n}{i}-(i+1)\binom{n-k}{i}+\binom{i+1}{2}\binom{n-2k}{i},\forall i \in \{n-3k+1,\dots, n-2k\}.\label{eq:Nk5}
\end{align}
\end{corollary}

Relying on these results we will analyze particular cases for a fixed $n$ and $k$ in particular ranges. These analyses will lead to simple formulae for the coefficients, and thus for the reliability of a consecutive system. Let us begin with $k\in \{1,2,n\}.$
\begin{proposition}~
\begin{itemize}
\item For $k=1$ $N_{n,1,i}=0, \forall i\neq n$, and $N_{n,k,n}=1$ and  \begin{equation}R(1,n;p)=p^n.\end{equation}  
\item For $k=2$ $N_{n,2,i}=\binom{i+1}{n-i}$ for $0\leq i\leq n$, and \begin{equation}R(2,n;p)=\sum\limits_{i=0}^n\binom{i+1}{n-i}p^iq^{n-i}.\end{equation}  
\item For $k=n$ $N_{n,n,i}=\binom{n}{i}, \forall i\ge 1$, and $N_{n,n,0}=0$ and \begin{equation}R(n,n;p)=\sum_{i=1}^n \binom{n}{i} p^i q^{n-i}.\end{equation}
\end{itemize}
\end{proposition}

Next, we consider the case when $n-2k<0$ in \eqref{eq:Nk4}, and the case when $n-3k<0$ in \eqref{eq:Nk5}.
\clearpage
\begin{proposition}\label{pr:2}~~

\begin{itemize}
\item For any $k\ge \lfloor\frac{n}{2}\rfloor$ we have $N_{n,k,i}=\binom{n}{i}, \forall i>n-k$, and $N_{n,k,i}=\binom{n}{i}-(i+1)\binom{n-k}{i}, \forall i\in\{\frac{n-k+1}{k},\dots,n-k+1\}.$ It follows that
\begin{equation} R(k,n;p)=\sum\limits_{i=i_{n,k}+1}^{n}\binom{n}{i}p^iq^{n-i}-\sum\limits_{i=i_{n,k}+1}^{n-k}(i+1)\binom{n-k}{i}p^iq^{n-i}.
\end{equation}
\item For any $\lfloor\frac{n}{3}\rfloor\le k< \lfloor\frac{n}{2}\rfloor$ we have $N_{n,k,i}=\binom{n}{i}, \forall i>n-k$, $N_{n,k,i}=\binom{n}{i}-(i+1)\binom{n-k}{i}, \forall i\in\{n-2k+1,\dots,n-k\}$, and $N_{n,k,i}=\binom{n}{i}-(i+1)\binom{n-k}{i}+\binom{i+1}{2}\binom{n-2k}{i}, \forall i\in\{\frac{n-k+1}{k},\dots,n-2k\}.$ This implies
\begin{align}
\begin{split}
     R(k,n;p)&=\sum\limits_{i=i_{n,k}+1}^{n}\binom{n}{i}p^iq^{n-i}-\sum\limits_{i=i_{n,k}+1}^{n-k}(i+1)\binom{n-k}{i}p^iq^{n-i}\\ &+\sum\limits_{i=i_{n,k}+1}^{n-2k}\binom{i+1}{2}\binom{n-2k}{i}p^iq^{n-i} .
\end{split}
     \end{align}
\end{itemize}
\end{proposition}

We now use the fact that $N_{n,k,i}$ can be efficiently computed for any $n$ and $k$ when $i\ge \lfloor\frac{n}{3}\rfloor$, to establish new bounds on the remaining coefficients. 
\begin{proposition}
  For any $k<\lfloor \frac{n}{3}\rfloor$ and $\forall i \in \{i_{n,k}+1,\dots,n-3k\}$
\begin{align}
\begin{split}
N_{n,k,i}&\leq \binom{n}{i}-(i+1)\binom{n-k}{i}+\binom{i+1}{2}\binom{n-2k}{i}\\
N_{n,k,i}&\geq \binom{n}{i}-(i+1)\binom{n-k}{i}.
\end{split}
\end{align}
\end{proposition}

Straightforward, we now define for any $k<\lfloor \frac{n}{3}\rfloor$ and $\forall i \in \{i_{n,k}+1,\dots,n-3k\}$ the upper and lower bounds as
\begin{align}
U_{n,k,i}&\triangleq\min\left\{\binom{n}{i}, \binom{n}{i}-(i+1)\binom{n-k}{i}+\binom{i+1}{2}\binom{n-2k}{i}\right\}\\
L_{n,k,i}&\triangleq\max \left\{0,\binom{n}{i}-(i+1)\binom{n-k}{i}\right\}.
\end{align}


\section{Simulations}

We have performed a series of simulations to test our results. We illustrate here only a small part of those, more exactly for $n\in\{16, 32, 64\}.$ 

In Figure \ref{fig:1} we plot $R(k,n;p)$ (i.e., $R(k,16;p)$ (\ref{fig:1a}), $R(k,32;p)$ (\ref{fig:1c}), and $R(k,64;p)$ (\ref{fig:1e})), as well as the relative errors of the approximation of $N_{n,k,i}$ using $L_{n,k,i}$ and $U_{n,k,i}$ in Figs. (\ref{fig:1b}), (\ref{fig:1d}), and (\ref{fig:1f}). More precisely, we plot $1-L_{n,k,i}/N_{n,k,i}$ for $k\ge \lfloor n/2\rfloor$ (light magenta) and $2<k<\lfloor n/3\rfloor$, and $1-U_{n,k,i}/N_{n,k,i}$ for $\lfloor n/3\rfloor \leq k<\lfloor n/2\rfloor$ (dark magenta).

\begin{figure}[htbp]
\centering
\subfloat[\label{fig:1a}]{\includegraphics[height=0.45\textwidth,width=0.45\textwidth]{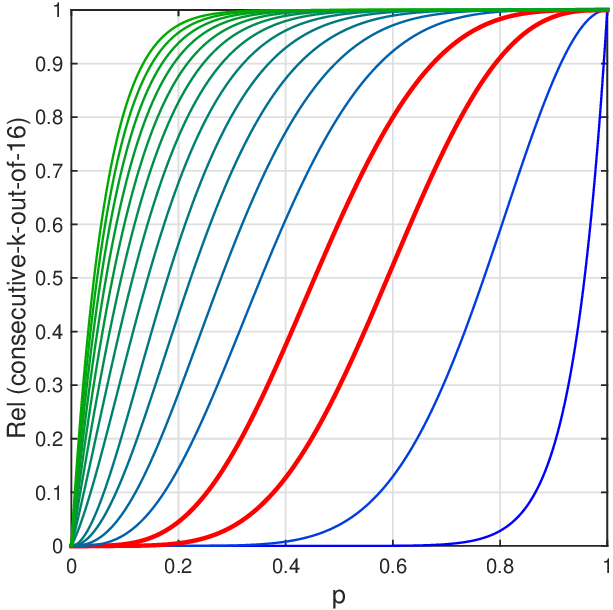}}\hfill
\subfloat[\label{fig:1b}] {\includegraphics[height=0.45\textwidth,width=0.45\textwidth]{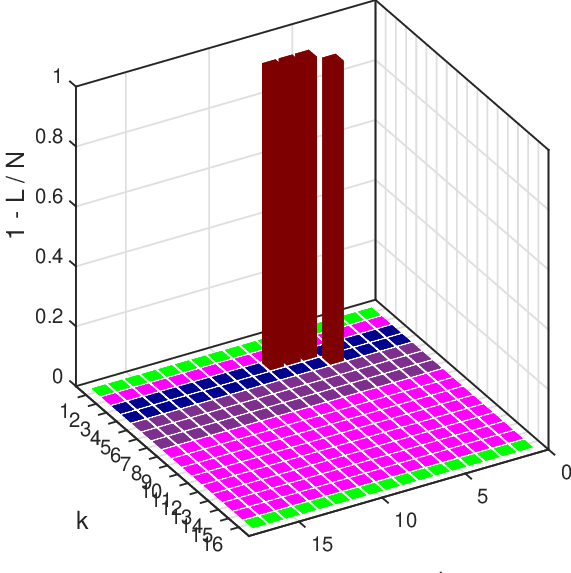}}\hfill
\subfloat[\label{fig:1c}]{\includegraphics[height=0.45\textwidth,width=0.45\textwidth]{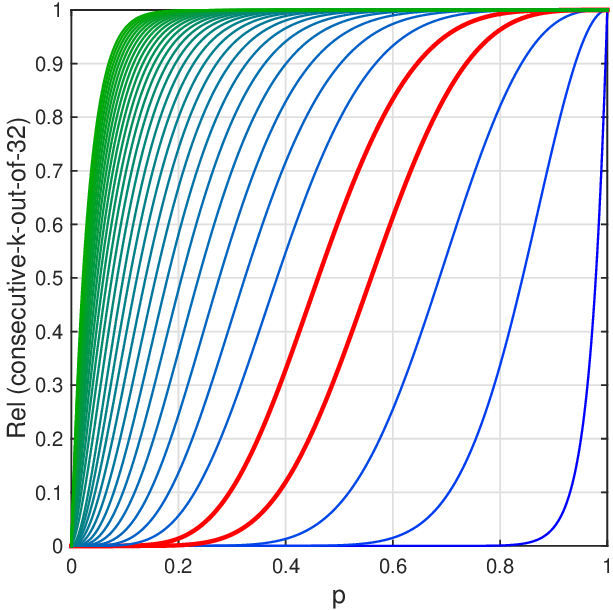}}\hfill
\subfloat[\label{fig:1d}] {\includegraphics[height=0.45\textwidth,width=0.45\textwidth]{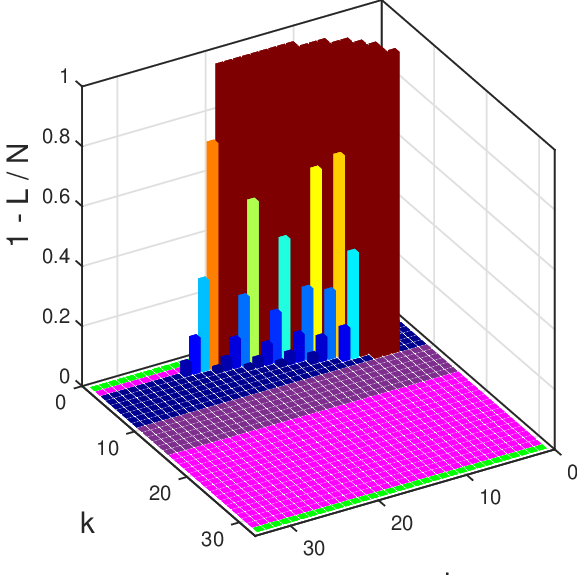}}\hfill
\subfloat[\label{fig:1e}]{\includegraphics[height=0.45\textwidth,width=0.45\textwidth]{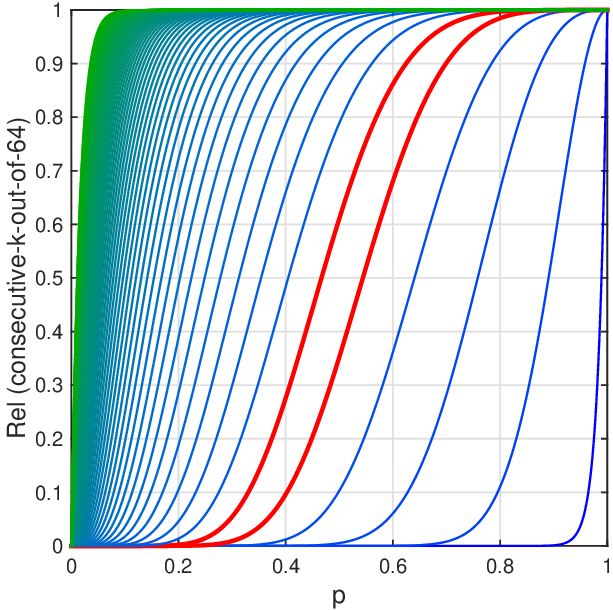}}\hfill
\subfloat[\label{fig:1f}] {\includegraphics[height=0.45\textwidth,width=0.45\textwidth]{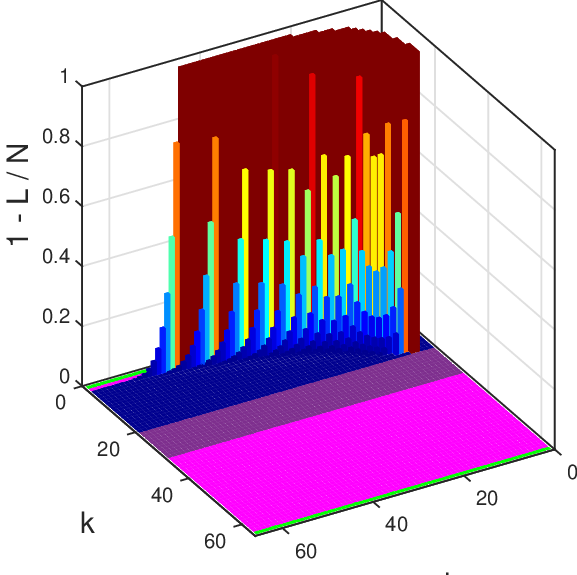}}\hfill
\caption{$R(k,n;p)$ for: (a) $n=16$, (c) $n=32$, and (e) $n=64$, as well as the relative errors for: $N_{n,k,i}$ for (b) $n=16$, (d) $n=32$, and (f) $n=64$.} \label{fig:1}
\end{figure}

\begin{paragraph}{Remarks}

\begin{itemize}
    \item The flat surfaces in Figs. \ref{fig:1b}, \ref{fig:1d}, and \ref{fig:1f} (green and magenta), show that the coefficients $N_{n,k,i}$ are computed exactly. This is a direct consequence of Proposition \ref{pr:2}.
    \item Focusing our attention on the case $2<k<\lfloor n/3\rfloor$, the absolute errors are different than $0$ in only a few cases.
    \item  The number of coefficients which are computed exactly (dark blue) is significantly larger than the number of approximated coefficients, e.g., for $n=32$, almost 81\% are computed exactly ($187$ out of $231$).
    \item The number of approximated coefficients is a decreasing function of $k$. Hence, as $k$ is approaching $n/3$, the number of exactly computed coefficients increases. For example, for $n=32$ and $k=9$ slightly over 90\% of the coefficients are computed exactly ($30$ out of $33$).  
    \item The worst approximation with respect to the absolute error ($N_{n,k,i}-L_{n,k,i}$) is achieved for $k=3$, and any $n\leq 64.$
\end{itemize}

\end{paragraph}

That is why we have decided to plot the exact reliability polynomial (red) together with the reliability polynomials obtained using the upper (green) and the lower (blue) bounds for $k=3$ and $n=16$ (Fig. \ref{fig:2}). Notice in Fig. (\ref{fig:2a}) that from $p\ge 0.5$ the approximations are practically overlapping with the exact reliability, while for smaller values of $p$ the behaviour of the two bounds can be seen in Fig. (\ref{fig:2b}).

\begin{figure}[htbp]
\centering
\subfloat[\label{fig:2a}]{\includegraphics[height=0.45\textwidth,width=0.45\textwidth]{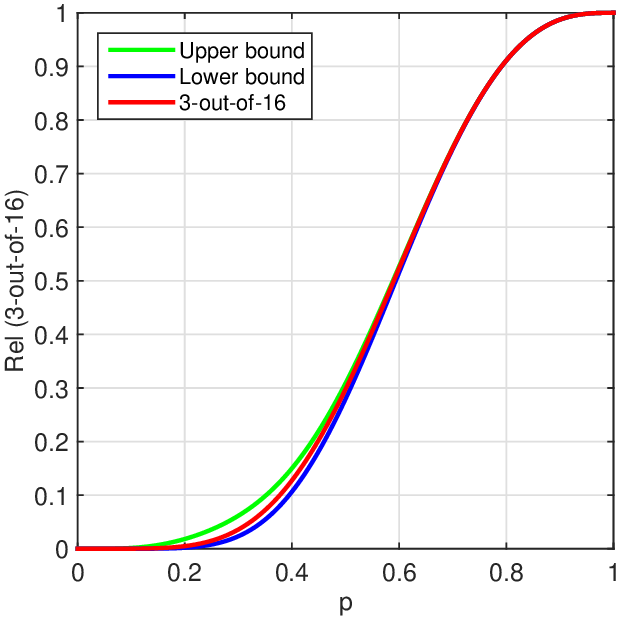}}\hfill
\subfloat[\label{fig:2b}] {\includegraphics[height=0.45\textwidth,width=0.45\textwidth]{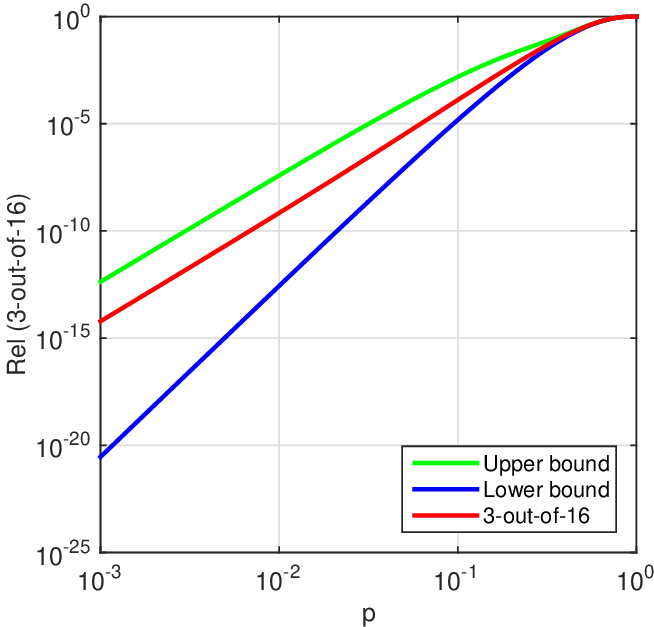}}\hfill
\caption{Reliability of a consecutive-$3$-out-of-$16$:$F$ system, and its upper and lower bounds: (a) linear scale, and (b) logarithmic scale.} \label{fig:2}
\end{figure}

\begin{figure}[htbp]
\centering
\subfloat[\label{fig:3a}]{\includegraphics[height=0.35\textwidth,width=0.45\textwidth]{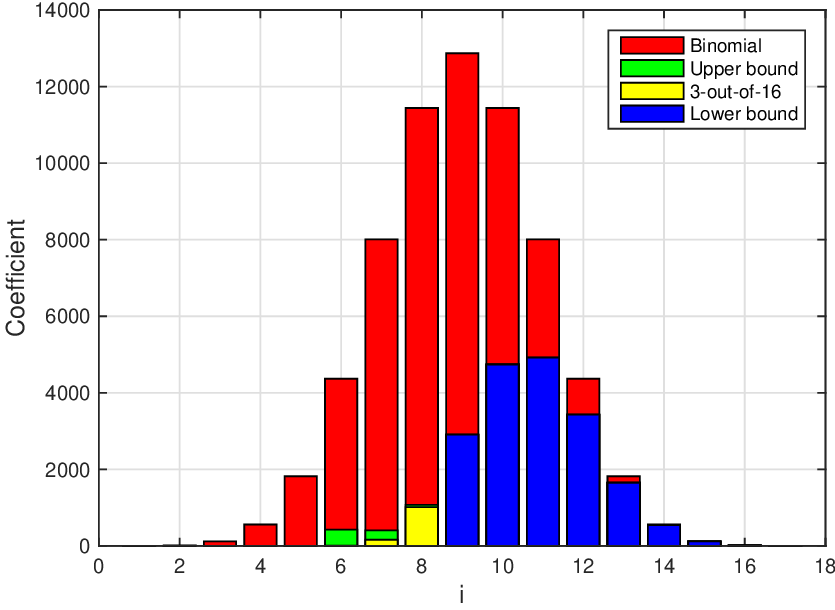}}\hfill
\subfloat[\label{fig:3b}] {\includegraphics[height=0.35\textwidth,width=0.45\textwidth]{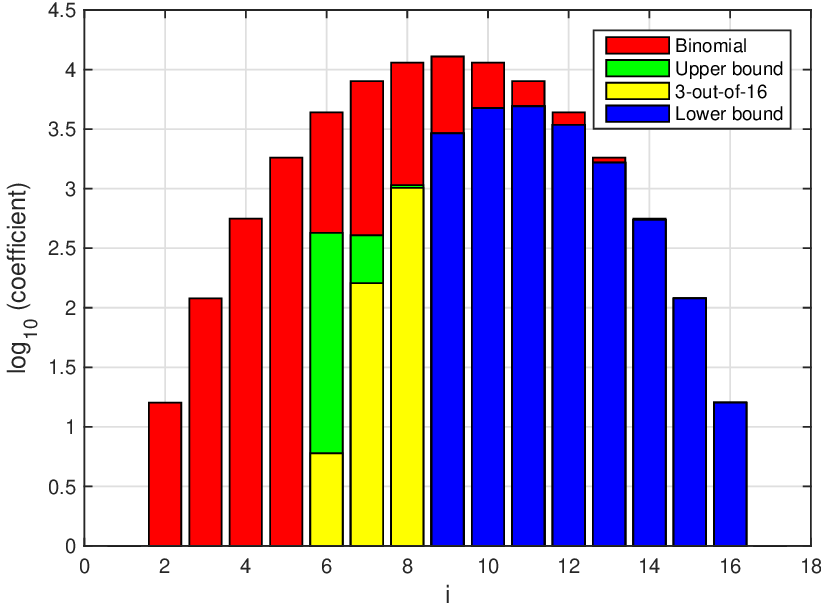}}\hfill
\subfloat[\label{fig:3d}]{\includegraphics[height=0.35\textwidth,width=0.45\textwidth]{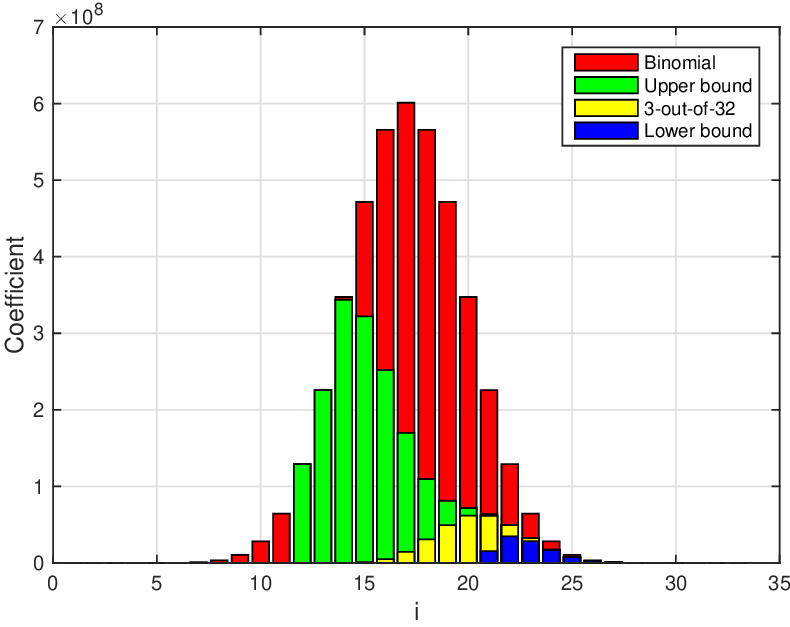}}\hfill
\subfloat[\label{fig:3e}] {\includegraphics[height=0.34\textwidth,width=0.45\textwidth]{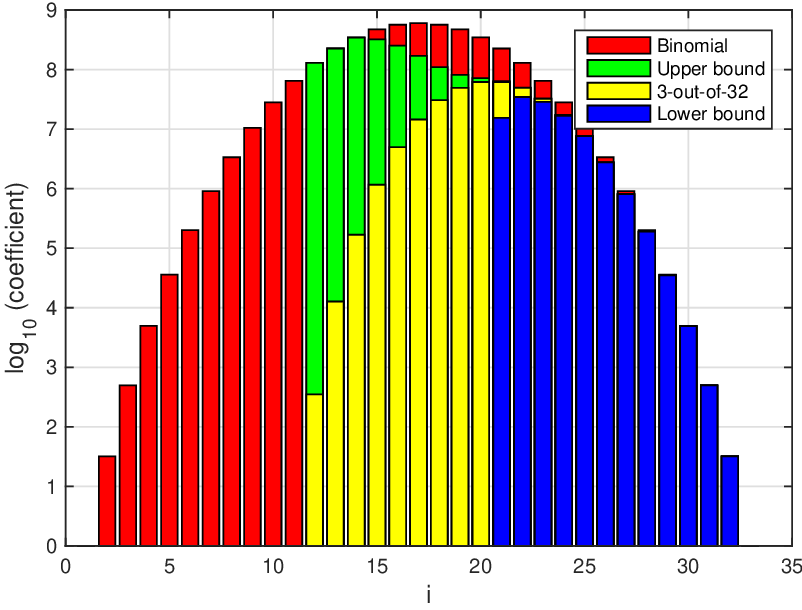}}\hfill
\subfloat[\label{fig:3f}]{\includegraphics[height=0.35\textwidth,width=0.45\textwidth]{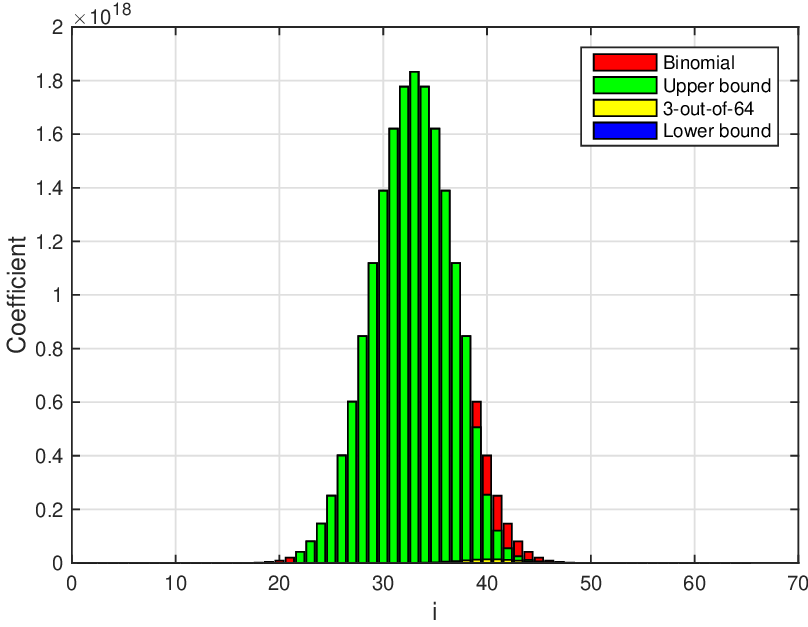}}\hfill
\subfloat[\label{fig:3g}] {\includegraphics[height=0.34\textwidth,width=0.45\textwidth]{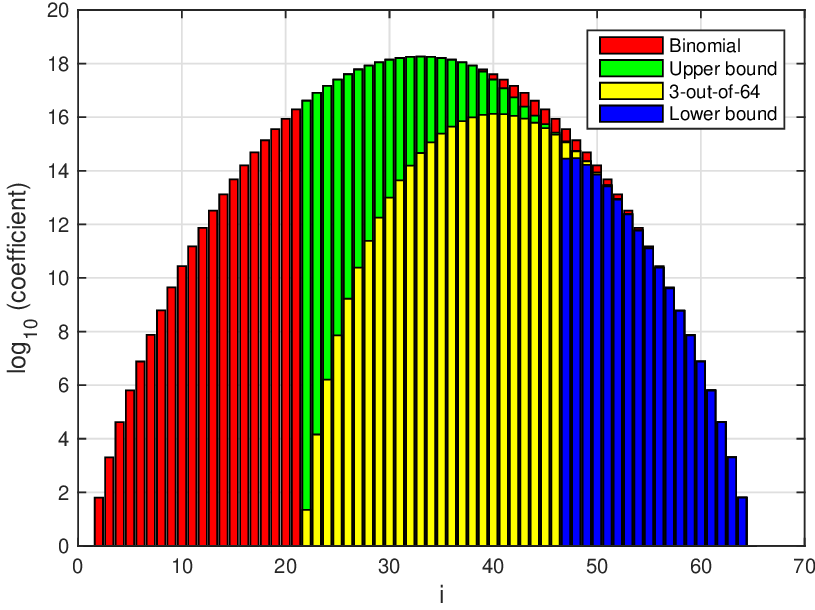}}\hfill
\caption{$N_{n,k,i}, L_{n,k,i}, U_{n,k,i}$, and $\binom{n}{i}$ in linear scale (left) and logarithmic scale (right) for: (a)-(b) $n=16$, (c)-(d) $n=32$, and (e)-(f) $n=64$.} \label{fig:3}
\end{figure}

Finally, Fig. \ref{fig:3} details the exact coefficients (yellow), as well as their lower (blue) and upper (green) bounds, on top of the corresponding binomial coefficients (red), in both linear and logarithmic scales.

\section{Conclusions}

In this paper, we have determined closed formulae for the reliability of a consecutive $k$-out-of-$n$:$F$ system expressed in the Bernstein basis. Based on the properties of the coefficients, we have proposed simple and easy to compute formulae for all $k\ge \lfloor n/3\rfloor$. For the remaining range of values, namely for $3\leq k<\lfloor n/3\rfloor$, we have proposed lower and upper bounds on the coefficients, and thus bounds on reliability. These bounds have several interesting properties, becoming sharper and sharper as $n$ gets larger, while requiring lower and lower computation work factors. 

The approach we have presented here opens the road to a new research direction in the area of consecutive systems. To our knowledge this is the first time bounding/approximating techniques have been used selectively only on a few of the coefficients of a consecutive system, rather than bounding the reliability polynomial. Detailed estimates of the trade-offs between computation complexity and accuracy of approximations have to be evaluated against previously published results for a better understanding of the advantages and disadvantages of the proposed approach (not included due to space limitations).    


\section*{Acknowledgements}
This research was supported by the European Union through the European Regional Development Fund under the Competitiveness Operational Program (\textit{BioCell-NanoART = Novel Bio-inspired Cellular Nano-Architectures}, POC-A1.1.4-E-2015 nr. 30/01.09.2016).

%
%
%
%
%
\bibliographystyle{plain}

\end{document}